\documentclass[journal]{IEEEtran}

\usepackage{cite}
\usepackage{amsmath}
\usepackage{amsthm}
\usepackage{amssymb}
\usepackage{mathtools}
\usepackage{float}
\usepackage{xcolor}
\usepackage{dblfloatfix}
\usepackage{url}
\usepackage{pdfpages}
\usepackage{siunitx}
\usepackage{hyperref}

\usepackage{tikz}
\usetikzlibrary{quantikz2}
\newcommand{\qvdots}{
  \raisebox{-.5em}[0pt][.25em]{\ensuremath{\vdots}}
}
\newcommand{\qddots}{
    \raisebox{0em}[1.5em][0pt]{\ensuremath{\ddots}}
}
\newtheorem{theorem}{Theorem}
\newtheorem{lemma}{Lemma}

\makeatletter
\newcommand{\beginsupplement}{%
  \clearpage

  \setcounter{section}{0}%
  \setcounter{subsection}{0}%
  \setcounter{subsubsection}{0}%
  \setcounter{equation}{0}%
  \setcounter{figure}{0}%
  \setcounter{table}{0}%

  \setcounter{theorem}{0}%
  \setcounter{lemma}{0}%

  \renewcommand{\thesection}{S\arabic{section}}%
  \renewcommand{\thesubsection}{S\arabic{section}.\arabic{subsection}}%
  \renewcommand{\thefigure}{S\arabic{figure}}%
  \renewcommand{\thetable}{S\arabic{table}}%
  \renewcommand{\theequation}{S\arabic{equation}}%
}
\makeatother

\makeatletter
\let\IEEEorig@maketitle\@maketitle
\let\IEEEorigthanks\thanks

\newcommand{\IEEEsecondmaketitle}{
  \par\begingroup
  \normalfont
  \def\thefootnote{}%
  \def\footnotemark{}%
  \let\@makefnmark\relax%
  \footnotesize
  \footnotesep 0.7\baselineskip
  \normalsize
  \if@twocolumn
    \twocolumn[{\IEEEquantizevspace{\IEEEorig@maketitle}[\IEEEquantizedisabletitlecmds]{0pt}[-\topskip]{\baselineskip}{\@IEEENORMtitlevspace}{\@IEEEMINtitlevspace}\@IEEEaftertitletext}]%
  \else
    \newpage\global\@topnum\z@ \IEEEorig@maketitle\@IEEEstatictitlevskip\@IEEEaftertitletext%
  \fi
  \thispagestyle{IEEEtitlepagestyle}\@thanks%
  \endgroup
  \setcounter{footnote}{0}%
  \gdef\@thanks{}%
  \let\thanks\relax%
}
\makeatother

\makeatletter

\def\ps@IEEEtitlepagestyle{%
  \def\@oddhead{\mycopyrightnotice}%
  \def\@evenhead{}%
}

\def\mycopyrightnotice{%
  \parbox{\textwidth}{%
    \scriptsize     
    \copyright~2026 IEEE. Personal use of this material is permitted. Permission from IEEE must be obtained for all other uses, including reprinting/republishing this material for advertising or promotional purposes, collecting new collected works for resale or redistribution to servers or lists, or reuse of any copyrighted component of this work in other works. %
  }%
}

\begin{document}

\title{Efficient Quantum Circuits for \\ the Hilbert Transform}

\author{Henry~Zhang and Joseph Li
\thanks{H. Zhang is with Jericho Senior High School, NY 11753, USA. (e-mail: henry-yunheng.zhang@jerichoschools.org).}
\thanks{J. Li is with the Department of Computer Science and the Joint Center for Quantum Information and Computer Science (QuICS), University of Maryland, College Park, MD 20742, USA (e-mail: jli0108@umd.edu).}}

\IEEEaftertitletext{\vspace{-1.05\baselineskip}}
\maketitle

\begin{abstract}
The quantum Fourier transform and quantum wavelet transform have been cornerstones of quantum information processing. However, for non-stationary signals and anomaly detection, the Hilbert transform can be a more powerful tool, yet no prior work has provided efficient quantum implementations for the discrete Hilbert transform. This letter presents a novel construction for a quantum Hilbert transform in polylogarithmic size and logarithmic depth for a signal of length $N$, exponentially fewer operations than classical algorithms for the same mapping. We generalize this algorithm to create any $d$-dimensional Hilbert transform in depth $O(d\log N)$. Simulations demonstrate effectiveness for tasks such as power systems control and image processing, with exact agreement with classical results.
\end{abstract}


\IEEEpeerreviewmaketitle

\section{Introduction}

\IEEEPARstart{O}{ver} the past few decades, the demand for high-performance computing has ballooned alongside exponentially growing datasets. Recent advances in machine learning and large-scale, high-precision simulations have exacerbated this trend as computational costs now reach exascales. However, at the same time, the transistors underpinning classical computers are now reaching atomic scales, so the era of predictable increases in computational power described by Moore's law appears to be coming to an end \cite{Shalf2020Future}.

Quantum computing has emerged as a solution to efficiently store and process high-dimensional data. Qubits, the unit of quantum information, can store $N$ data points using $\log_2 N$ qubits, and quantum entanglement and superposition allow for inherent parallelization. Two of the most important algorithms in signal processing, the discrete Fourier transform (DFT) and discrete wavelet transform, have well-described quantum analogues, known to be exponentially faster as a subroutine when compared to their classical counterparts \cite{Shor1994, Li2018QWT}. Growing literature has demonstrated their widespread utility in quantum algorithms for phase estimation, computational chemistry, information processing, and more \cite{Harrow2009HHL, Ouedrhiri2021ML, Dutta2024Denoising, Li2017Filtering}. Recent work has also explored quantum formulations of discrete trigonometric, wave atom, Mellin, and curvelet transforms \cite{Klappenecker2001DCT, Podzorova2025Waveatom, Twamley2006Mellin, Liu2009Curvelet}, yet one of the most fundamental operators in classical signal analysis, the Hilbert transform, remains essentially unexplored.

By generating a signal with phase orthogonal to its input, the Hilbert transform can be used for envelope detection, feature extraction, and phase alignment \cite{Liu2012Hilbert} in fields such as power systems control \cite{Derviskadic2020BeyondPhasors}, seismology \cite{Purves2014Phase}, and climate modeling \cite{Zappal2016Atmosphere, Rieger2021Climate}. Its extraction of instantaneous information from arbitrary signals can be more useful than the global estimation of the Fourier transform, and comes without the need for basis function selection of wavelet transforms. The two-dimensional Hilbert transform extends to image processing, where its properties can be used for edge detection, envelope detection, and corner detection \cite{LorenzoGinori2dHT, Kohlmann1996Corner}, and the three-dimensional Hilbert transform has applications including seismic imaging and microscopy \cite{Zhang20193d, Arnison20003d}. 

To compute the discrete Hilbert transform classically, the convolution theorem is invoked to represent the transformation as pointwise multiplication after the fast Fourier transform. However, deterministic pointwise multiplication of arbitrary quantum states is provably impossible because of the linearity of quantum mechanics \cite{Lomont2003Convolution}, and arbitrary convolutions on quantum computers achieve success probability at most $O(1/N)$ for length-$N$ state vectors \cite{Holmes2023Nonlinear,Pfeffer2024Comment}, rendering naïve approaches highly infeasible.

This letter presents an efficient quantum algorithm for the discrete Hilbert transform by exploiting the relatively structured nature of the Hilbert transform in the frequency domain: the removal of the zero-frequency DC component of the signal followed by quadrature phase shifts. We show that suppressing the DC mode can be achieved simply without the need for deterministic pointwise multiplication by isolating the zero-frequency amplitude onto an ancilla qubit using a sequence of multi-controlled-X gates and measurement. Moreover, the Hilbert transform's phase shifts are implemented in constant depth by applying Pauli-Z gates to the most significant qubit of each register, which precisely encodes the sign bit of each frequency component. In $d$ dimensions, our algorithm's depth scales logarithmically in $N$ and linearly in $d$ as $O(d \log N)$, compared to the $O(N^d \log N^d)$ scaling of the FFT, showing significant promise for multidimensional signal processing.

The remainder of this letter is structured as follows: Section II details our algorithm for the $d$-dimensional quantum Hilbert transform. Section III analyzes the time complexities of the algorithms. Section IV discusses potential applications for the one and two-dimensional quantum Hilbert transforms, and numerical comparisons to classical results. Finally, Section V concludes with potential directions for future research.

\section{Quantum Hilbert Transform}
\subsection{Preliminaries}
The Hilbert transform of a sufficiently smooth \cite{Arcos2021Hilbert} real function $f(t)$ is given by the integral
\begin{equation}
\label{eq:analytic}
\mathcal{H}f(t) = \frac{1}{\pi}\lim_{\epsilon \to 0^+} \int_{|\tau-t|>\epsilon} \frac{f(\tau)}{t-\tau} \text{d}\tau.
\end{equation}
This integral can be rewritten as the convolution $f * \frac{1}{\pi t}$ in the sense of tempered distributions, which, by the convolution theorem, has the frequency domain representation
\begin{equation}
\widehat{\mathcal{H}f}(\xi) = -i\kern0.05em\text{sgn}(\xi)\hat f(\xi).
\end{equation}
The factor $-i\kern0.05em\text{sgn}(\xi)$ acts as a frequency-domain filter that shifts the phase of positive frequencies by $-\pi/2$ radians and negative frequencies by $+\pi/2$ radians, while removing the DC component, which corresponds to when $\xi = 0$.

For a discrete $f \in \mathbb{R}^N$, where, without loss of generality $N = 2^n\in \mathbb{N}$, discrete Hilbert transform algorithms perform the analogous operation using the DFT:
\begin{equation}
\label{eq:dht}
\mathcal{H}[f]_k = \frac{1}{N} \sum_{w=0}^{N-1} -i\,\text{sgn}(w)\left( \sum_{j=0}^{N-1} f_j e^{-2\pi i j w / N} \right) e^{2\pi i k w / N},
\end{equation}
for frequency indices $w \in \{0,1,\ldots, N-1\}$. Note that although the continuous Hilbert transform on $L^2(\mathbb R)$ is a unitary operation, finite-$N$ discrete Hilbert transforms are not, because the DC bin of the DFT is projected out.

We consider the discrete Hilbert transform generalized to $d$ dimensions, which extends the same form along each dimension. Let $f \in \mathbb R^{N\times\cdots \times N}$ henceforth be a real, order-$d$ tensor. The multidimensional transform $\mathcal{H}[f]_{k_1,\ldots,k_d}$ is
\begin{equation}
\label{eq:d-dht}
\frac{1}{N^d} \sum_{w_1,\ldots,w_d=0}^{N-1} \prod_{m=1}^d\left(-i\text{sgn}(w_m)\right)\hat{f}_{w_1,\ldots,w_d}e^{2\pi i\sum_{j=1}^{d}k_jw_j/N},
\end{equation}
with Fourier coefficients
\begin{equation}
\hat{f}_{w_1,\ldots,w_d} = \sum_{j_1,\ldots, j_d=0}^{N-1}f_{j_1,\ldots,j_d}e^{-2\pi i\sum_{m=1}^dj_mw_m/N}.
\end{equation}

\subsection{Quantum Algorithm}

We construct quantum circuits for any $d$-dimensional discrete Hilbert transform, where the input tensor $f$ is Frobenius-normalized and represented as the quantum state
\begin{equation}
\ket{f} = \sum_{k_1,\ldots,k_d=0}^{N-1} f_{k_1,\ldots,k_d}\ket{k_1\cdots k_d},
\end{equation}
encoded in $d$ registers each of $n$ qubits, with $\ket{k_1\cdots k_d} = \ket{k_1} \otimes\cdots\otimes \ket{k_d}$. An ancilla, separate from the data registers, is prepended as the most significant qubit to make the initial state $\ket{0}\ket{f}$. Our algorithm is then as follows:

(i) First, an $n$-qubit QFT is applied in parallel on each of the $d$ registers. Because of the separability of the multidimensional Fourier transform \cite{Pfeffer2023QFT}, this yields
\begin{equation}
\frac{1}{N^{d/2}}\sum_{j_1,\ldots,j_d=0}^{N-1}\hat{f}_{j_1,\ldots,j_d}\ket{0}\ket{j_1\cdots j_d},
\end{equation}
with Fourier coefficients defined by
\begin{equation}
\label{eq:d-fouriercoeff}
\hat{f}_{j_1,\ldots,j_d} = \sum_{k_1,\ldots,k_d=0}^{N-1} f_{k_1,\ldots,k_d} e^{2\pi i \sum_{m=1}^d j_m k_m /N}.
\end{equation}

(ii) After the QFT, we wish to remove the DC component, which corresponds to any amplitude with at least one register in the basis state $\ket{0^n}$. To begin, for each register $r$, we apply a multi-controlled X (MCX) targeting the ancilla, with every qubit of $r$ as a negative control, which flips the ancilla if and only if the register equals $\ket{0^n}$. Then, immediately after each MCX, we measure the ancilla in the computational basis and reset it to $\ket{0}$. Conditioned upon the ancilla in $\ket{0}$ each time, the disjoint measurements have the effect:
\begin{equation}
    \Pi_{\text{keep}}
    = \bigotimes_{r=1}^d \bigl(I_r - \Pi^{(r)}_{0}\bigr)
    = \sum_{\substack{(j_1,\dots,j_d)\\ \forall r: j_r\neq 0}}
    \ket{j_1\cdots j_d}\bra{j_1\cdots j_d}.
\end{equation}
Applying $\Pi_\text{keep}$ to the state (then omitting the ancilla) yields
\begin{equation}
\label{eq:measure2}
\frac{1}{N^{d/2}\sqrt{1-p}}\sum_{\substack{(j_1,\dots,j_d) \\ \forall r:j_r \neq 0}} \hat f_{j_1,\ldots,j_d}\ket{j_1\cdots j_d},
\end{equation}
removing the DC bin and allowing us to work around the non-unitarity of the discrete Hilbert transform. Here, $1-p=\mathbb P[\text{success}] = \|\Pi_\text{keep}|\hat f\rangle\|^2$ is the success probability of the postselection; in Lemma 2 we show that $p$ equals the fraction of the spectral energy contained in the DC component.

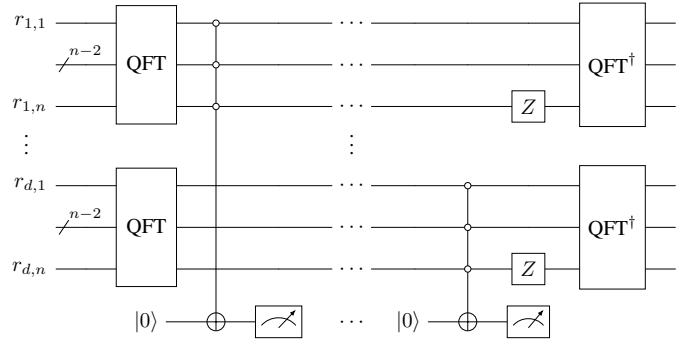
\begin{figure}[t]
    \hspace*{-.8em}
    \scalebox{0.8}{
    \begin{quantikz}[classical gap=0.06cm, row sep=0.8em, wire types={q,q,q,n,q,q,q,n}, thin lines]
        \lstick{$r_{1, 1}$} && \gate[3]{\kern0.08em\text{QFT}\kern0.08em} & \octrl{1} &&\ \cdots \ & &&&\gate[3]{\text{QFT}^\dagger}&\\
        &\qwbundle[style={xshift=-1.2mm}]{n-2}&& \octrl{1} && \ \cdots \ &&&&& \\
        \lstick{$r_{1, n}$} &&& \octrl{5} && \ \cdots \ &&&\gate{Z}&& \\
        \lstick{\qvdots\kern0.5em} & & & & & \push{\qvdots}&& \\
        \lstick{$r_{d, 1}$} && \gate[3]{\kern0.08em\text{QFT}\kern0.08em} &&&\ \cdots \ && \octrl{1} & & \gate[3]{\text{QFT}^\dagger} & \\
        &\qwbundle[style={xshift=-1.2mm}]{n-2} &&&&\ \cdots \ &&\octrl{1} & & &\\
        \lstick{$r_{d, n}$} &&&&&\ \cdots \ &&\octrl{1} & \gate{Z} & &\\
        && \ \push{{\ket{0}}}\ & \targ{} \setwiretype{q} &  \meter{} \hspace{6pt} & \setwiretype{n}  \ \cdots \ &  \ \hspace{-6pt} \push{{\ket{0}}}\  \setwiretype{n}& \targ{} \setwiretype{q} & \meter{}  &  \setwiretype{n} \push{\kern1.75em} \\
    \end{quantikz}
    }
    \caption{The circuit diagram of the $d$-dimensional quantum Hilbert transform. Each register is $n$ qubits, with an ancilla as the most significant qubit.}
    \label{fig:circuit}
\end{figure}

(iii) To implement the phase shifts of the Hilbert transform, we apply a Pauli‑Z gate on the most significant qubit of every register. Letting $b_r$ be the most significant bit of register $r$, the action of the Pauli-Z on $\ket{b_r}$ is $(-1)^{b_r}\ket{b_r}$. Applying a Pauli-Z on the most significant qubit of all $d$ registers therefore multiplies $\ket{j_1\cdots j_d}$ by 
\begin{equation}
\prod_{r=1}^d(-1)^{b_r} = (-1)^{\sum_{r=1}^d b_r} = (-1)^{\sum_{r=1}^d \left\lfloor \frac{j_r}{2^{n-1}}\right\rfloor},
\end{equation}
which is exactly the sign pattern of \eqref{eq:d-dht} under the standard frequency ordering of the DFT, where the first $N/2=2^{n-1}$ frequencies are nonnegative (DC and positive), and the next $2^{n-1}$ are negative. The state after this step is thus
\begin{equation}
\frac{1}{N^{d/2}\sqrt{1-p}}\sum_{\substack{(j_1,\dots,j_d) \\ \forall r:j_r \neq 0}} (-1)^{\sum_{r=1}^d \left\lfloor \frac{j_r}{2^{n-1}}\right\rfloor} \hat f_{j_1,\ldots,j_d}\ket{j_1\cdots j_d}.
\end{equation}

(iv) Finally, we employ an $n$-qubit inverse QFT in parallel on each of the $d$ registers to get
\begin{equation}
\ket{\mathcal{H}[f]} =\frac{1}{N^d\sqrt{1-p}}\sum_{k_1,\ldots,k_d=0}^{N-1}f'_{k_1,\ldots,k_d}\ket{k_1\cdots k_d}.
\end{equation}
This is the normalized, $d$-dimensionally Hilbert-transformed state up to an unobservable global phase $(-i)^d$, where $f'_{k_1,\ldots,k_d}$ is
\begin{equation}
\sum_{\substack{(j_1,\dots,j_d)\\
            \forall r:j_r\neq 0}}
(-1)^{\sum_{r=1}^{d}\left\lfloor \frac{j_r}{2^{n-1}}\right\rfloor}
\hat f_{j_1,\ldots ,j_d}
e^{-2\pi i\sum_{l=1}^{d} j_l k_l/N}.
\end{equation}

\section{Complexity Analysis}
We show that: (1) our quantum Hilbert transform (QHT) construction has logarithmic depth in terms of $N$ and uses only $dn+1$ qubits; (2) the acceptance probability of the postselection (success probability of step (ii)) is $1-p$; and (3), up to a quadratic end-to-end speedup can be obtained for estimating any component of the output.

\begin{lemma}
Let $f \in \mathbb R^{N\times\cdots\times N}$ be a real order-$d$ tensor. There exists a quantum circuit which maps $\ket{f} \to \ket{\mathcal{H}[f]}$ with circuit size $O(dn^2)$ and depth $O(dn)$, using $dn+1$ qubits.
\end{lemma}
\begin{proof}
    See the Supplemental Material.
\end{proof}

\begin{lemma}
Let $\ket{f} \in \mathbb C^{N^d}$ have Fourier coefficients $\hat{f}_{j_1,\ldots,j_d}$ defined in Equation \eqref{eq:d-fouriercoeff} and spectral energies $S_{j_1,\ldots,j_d} = |\hat{f}_{j_1,\ldots,j_d}|^2$. Suppose the DC component satisfies
\begin{equation}
    \sum_{\substack{(j_1,\dots,j_d) \\ \exists m:j_m =0}}S_{j_1,\ldots,j_d} = p\sum_{j_1,\ldots,j_d=0}^{N-1} S_{j_1,\ldots,j_d}
\end{equation}
for some $p \in [0, 1]$. Then, $\mathbb{P}[\textnormal{success}] = 1-p$.
\end{lemma}
\begin{proof}
    See the Supplemental Material. 
\end{proof}

\begin{theorem}
    Let $f \in \mathbb R^{N\times\cdots\times N}$ be a real order-$d$ tensor. There exists a quantum algorithm which estimates a chosen component of $|\mathcal H[f]|^2$ within an additive error $\epsilon \in (0, 1)$ with an end-to-end time complexity of $O(nd/(\epsilon^2(1-p)))$ using sampling, or $O(nd/\epsilon)$ with amplitude estimation.
\end{theorem}
\begin{proof}
In this letter, we follow the assumption that the state preparation is efficient, i.e., it can be achieved in polylogarithmic depth with the aid of ancillas \cite{Zhang2022QSP, Sun2023QSP}, it is easily preparable, like a matrix product state \cite{PhysRevLett.95.110503}, or, most likely, it is already provided given that the QHT is used as a subroutine for a larger quantum algorithm.

By Lemma 1, the depth of our quantum circuits are $O(dn)$, with success probability $1-p$ by Lemma 2 (with $p$ dependent upon the DC component of the input state), meaning we expect $1/(1-p)$ trials until the first success. Sampling trials to estimate probabilities scales as $O(1/\epsilon^2)$ by Chebyshev's inequality. Hence, estimating a value of $|\mathcal H[f]|^2$ via sampling requires time $O(nd/(\epsilon^2(1-p)))$.

For a better asymptotic bound, we can leverage amplitude estimation \cite{Brassard2002}, embedding the QHT into a Grover-like operator so that estimating the amplitude to $\epsilon$ requires $O(1/\epsilon)$ calls to the subroutine rather than $O(1/\epsilon^2)$, while additionally removing the need for postselecting $\ket{0}$ as in Lemma 2. This gives the time complexity $O(nd/\epsilon)$, although requiring up to $d-1$ more ancillas for the same depth because the ancillas cannot be reset and reused with mid-circuit measurements.
\end{proof}

Using amplitude estimation gives up to a quadratic speedup compared to the fastest classical algorithms, while maintaining the reduced discretization errors that come from being able to analyze a larger dataset. For a $d$-dimensional discrete Hilbert transform, the fastest known classical algorithms for computing the entire output $\mathcal H[f]$ have size $O(N^d \log N)$ using parallel fast Fourier transforms \cite{Liu1981Numerical} or similar algorithms for discrete trigonometric transforms, which share the same time complexity \cite{Bilato2014Fast, Micchelli2011Spline}. For estimating $k \ll N^d$ components of $|\mathcal H[f]|$, the fastest known classical cost remains $O(kN^d)$, by directly computing the discrete convolution sum \cite{Liu1981Numerical}.

In comparison, our quantum end-to-end cost to estimate $k$ amplitudes is $O(kd\log N/\epsilon)$. For fixed $\epsilon$, this is exponentially faster in $N$, and even for $\epsilon$ growing proportionally to amplitudes, we obtain a quadratic speedup over classical methods. We emphasize that the QHT always maps $\ket{f}\to\ket{\mathcal{H}[f]}$ in logarithmic depth and polylogarithmic size, which is exponentially fewer operations than the classical implementation of the Hilbert transform. Consequently, the QHT may also yield exponential speedups in quantum algorithms that use the Hilbert transform as a subroutine, but do not require readout of all components of the output vector $\ket{\mathcal{H}[f]}$.

\begin{table}[]
\caption{Asymptotic complexities of classical discrete Hilbert transforms and our quantum approach.}
\begin{tabular}{l|l|l}
           & {Classical} (best known) & {Quantum} (this work) \\\hline
Transform cost  & $O(N^d \log N)$          & $O(d \log^2 N)$      \\
End-to-end cost & $O(\min(N^d\log N^d, kN^d))$    &     $O(kd\log N/\epsilon)$         \\  
\end{tabular}
\end{table}

\section{Numerical Results}

The efficacy of multidimensional QHTs is tested and verified in three cases using the IBM Qiskit simulator.
\begin{figure*}[!t]
    \centering
    \addtocounter{figure}{1}
    \includegraphics[width=\linewidth]{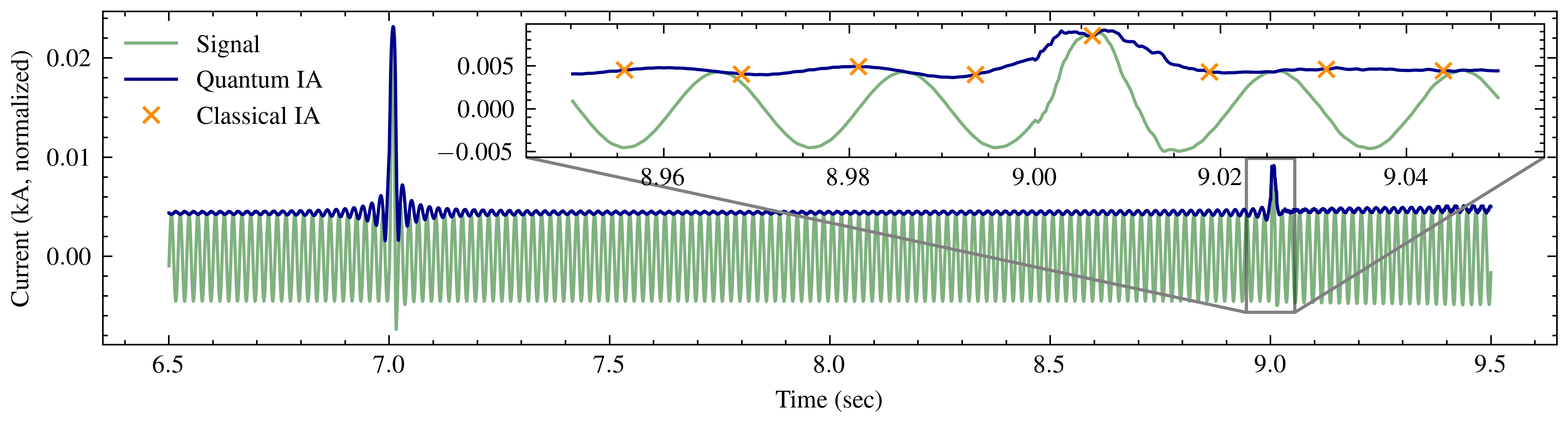}
    \caption{The quantum Hilbert transform precisely captures fast multi-fault events in HVDC grids.} % for efficient wide-area protection and control of power grids.}
    \label{fig:power-fig}
\end{figure*}
\begin{figure}[H]
    \centering
    \addtocounter{figure}{-2}
    \includegraphics[width=\linewidth]{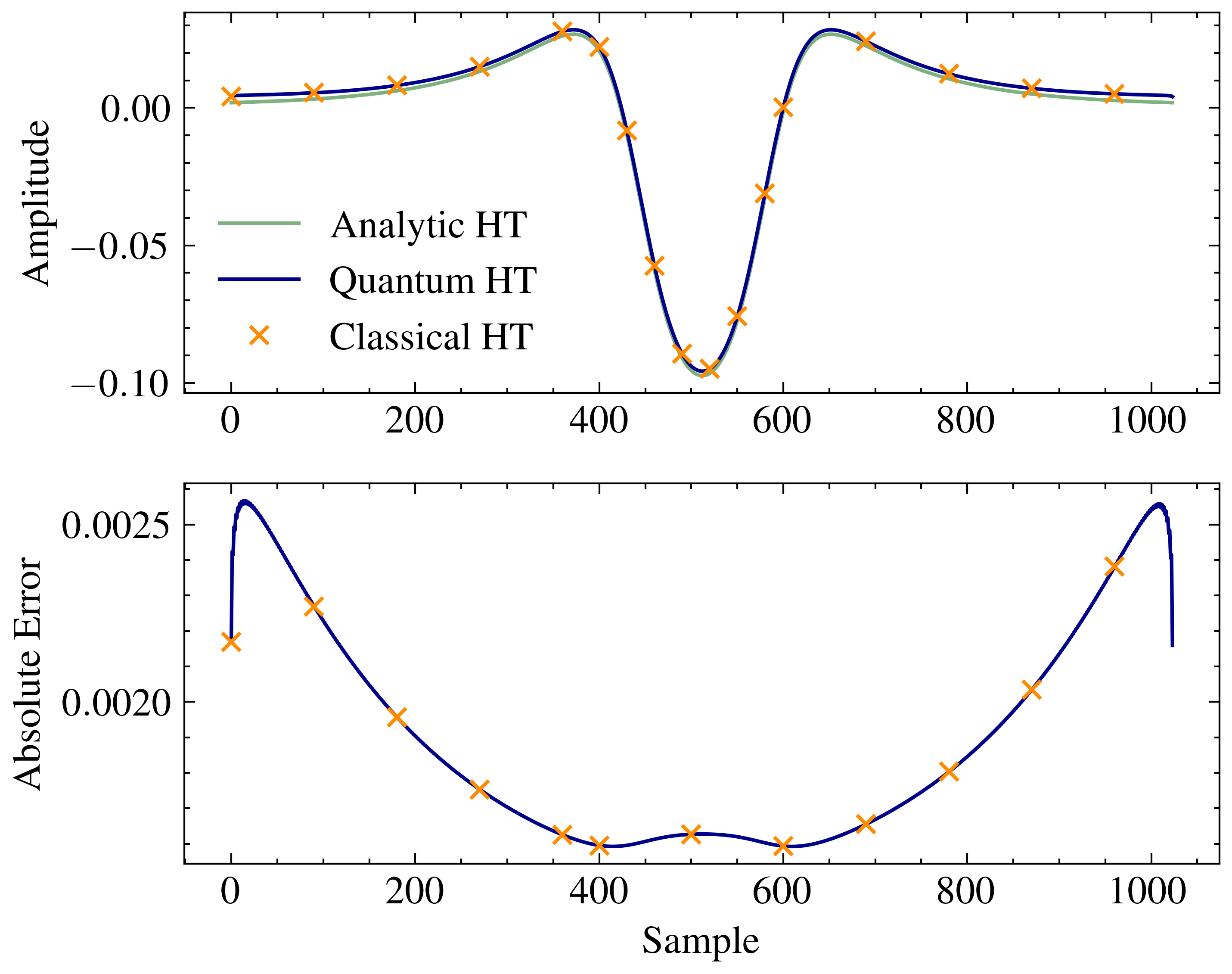} 
    \caption{The quantum Hilbert transform agrees exactly with classical results calculated using the fast Fourier transform for $f(x) = \sin(x)/(1+x^4)$.}
    \label{fig:analytic-fig}
\end{figure}

\subsection{Methodology Validity}
In Figure \ref{fig:analytic-fig}, we choose the standard test case $f(x) = \sin(x)/(1+x^4)$ from \cite{Bilato2014Fast}, for its combination of oscillatory and decay behavior, with the known analytic Hilbert transform
\begin{equation}
    \mathcal Hf(x) = \frac{e^{\frac{-1}{\sqrt{2}}} \cos \frac{1}{\sqrt{2}}+e^{\frac{-1}{\sqrt{2}}} \sin \frac{1}{\sqrt{2}} x^2-\cos x}{1+x^4}
\end{equation}
as the ground truth. We sample $f(x)$ at $N=128$ points uniformly in an interval centered about $x=0$ with step size $h=0.01$, and compute the discrete transform via both FFT and a QHT circuit on 7 data qubits plus one ancilla. 

To quantify the agreement between results, we compute the fidelity $F= |\braket{\psi_\text{ref}}{\mathcal H[f]}|^2$, a standard metric in quantum information theory, where $\psi_\text{ref}$ is the classically prepared reference, normalized for comparison. The QHT output fidelity agrees with the classical FFT-based result to floating-point accuracy, i.e., $F>1-1.5\times10^{-10}$, demonstrating perfect correspondence.

\subsection{Capturing Fast Transients}
Figure \ref{fig:power-fig} applies the QHT to capturing the dynamic phasor information of fault currents and voltages described by \cite{Derviskadic2020BeyondPhasors}. This process is essential for fault identification, localization, and mitigation in modern power grids to prevent blackouts. We demonstrate the QHT's capability on amplitude-encoded data from a large 9‑terminal HVDC-380 kV AC grid \cite{jovcic2021dc}, in a scenario with two pole-to-ground faults occurring at two terminals of DC converters subsequently at $t\approx\SI{7}{\second}$ and $t\approx\SI{9}{\second}$ and lasting \SI{2}{\milli\second} and \SI{5}{\milli\second} respectively. The instantaneous amplitude (IA), extracted through the discrete Hilbert transform, reveals abrupt deviations in both current waveforms that coincide with the fault onset times, enabling accurate temporal and spatial localization of fault events.

Remarkably, with one extra ancilla used for recursive MCX decomposition \cite{Barenco1995Gates}, the quantum circuit required for computing the instantaneous amplitude of the phasor data in this small-scale test transpiles to just 1,565 single-qubit rotations and CNOTs, compared to at least 4,915,200 floating-point operations required to compute the Cooley-Tukey FFT and inverse FFT. The QHT result once again matches the FFT-based result to floating-point accuracy. The raw data for this case can be found in the IEEE data portal \cite{9xeg-4g78-25}. 
 
\subsection{Corner Detection in 2D} 
Finally, we verify the efficacy of two-dimensional QHT in a corner detection case. We encode a normalized, $1024\times 1024$, grayscale chessboard image $f_{x, y}$ for $x, y \in \{ 0, 1,\ldots,1023\}$ into the amplitudes of a quantum state on $n=20$ data qubits plus two ancillas. As shown in Figure \ref{fig:2d-fig}, the phase effect of the QHT correctly identifies all 64 corners of the chessboard as described by   \cite{Kohlmann1996Corner, Ziou1998EdgeDetection}.
\begin{figure}[h]
    \centering
    \addtocounter{figure}{1}
    \includegraphics[width=\linewidth]{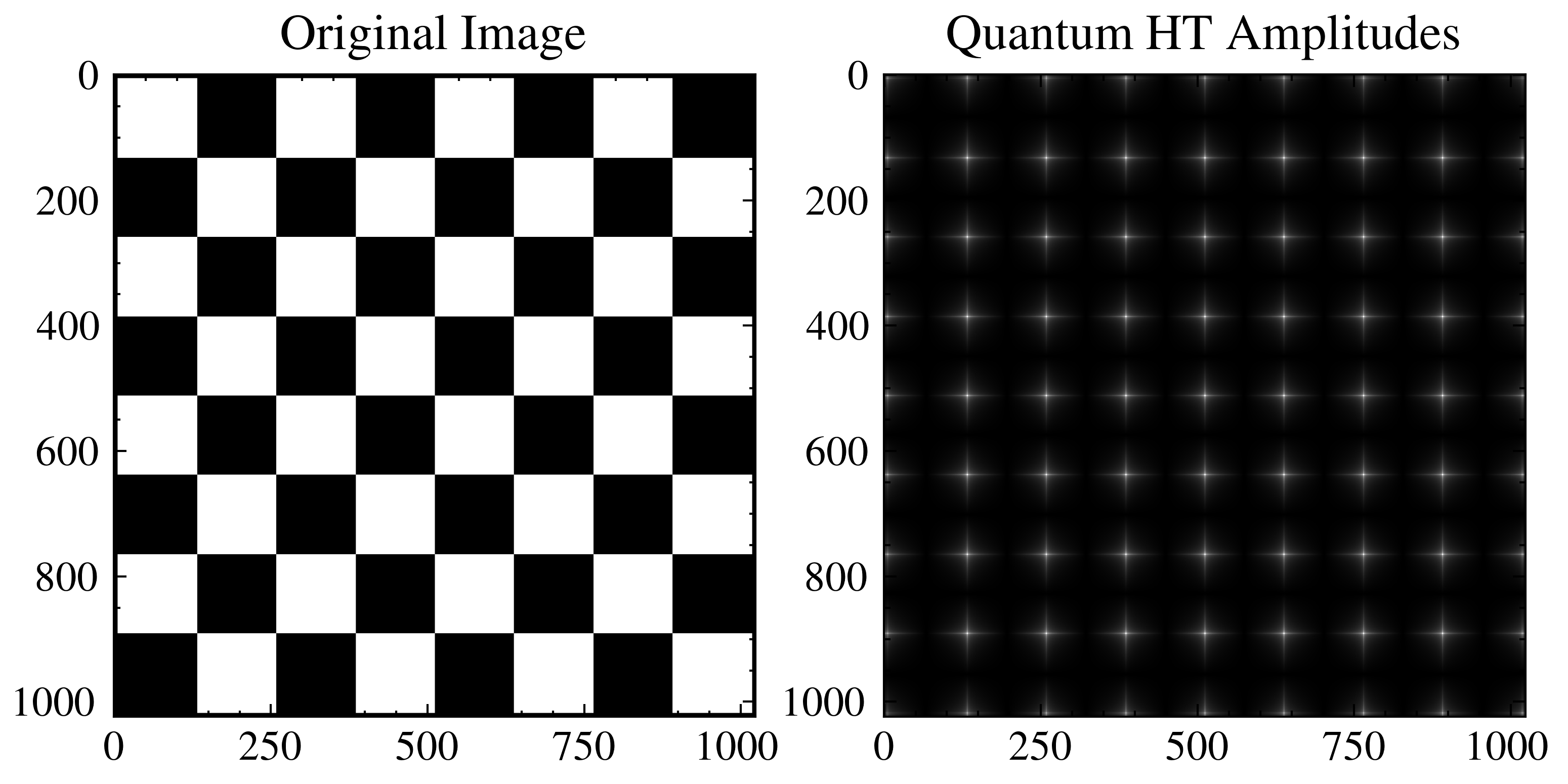}
    \caption{The absolute amplitudes of the quantum Hilbert transform detect the corners of a chessboard.}
    \label{fig:2d-fig}
 \end{figure}

In total, the simulated circuit for this test case used 1,590 rotation and CNOT primitives. In contrast, a classical separable 2D FFT plus quadrature filtering on a \(1024\times 1024\) array requires approximately $2.1\times10^8$ real operations. This reduction in required operations highlights the promise of the quantum Hilbert transform for multidimensional processing tasks.

\section{Conclusion}

We presented a novel quantum algorithm that implements the discrete Hilbert transform in any fixed dimension $d$ using only $O(\log^2 N)$-size and $O(\log N)$-depth circuits, by creating efficient subroutines for quantum DC filtering and quadrature phase shifts without arbitrary state multiplication. The proposed algorithm was validated through simulations on both one-dimensional analytic benchmarks and one- and two-dimensional processing tasks, where the quantum implementation achieves a substantial reduction in circuit complexity, requiring orders of magnitude fewer operations compared to classical FFT-based Hilbert transforms. Future work might investigate the integration of the QHT into data-intensive applications such as seismic signal analysis or dynamic protection schemes for power grids, and test its robustness on intermediate-scale quantum hardware.

\textit{Note added.} A paper by Jha and Parakh \cite{Jha2025Sigh} appeared on arXiv while this work was in the final stages of preparation. Their discussion introduces the term \textit{quantum Hilbert transform} but does not develop any circuit implementation or resource analyses. Our results were developed independently.

\bibliographystyle{IEEEtran}
\bibliography{ref}

\newpage 
\beginsupplement

\makeatletter\let\thanks\IEEEorigthanks\makeatother

\title{Supplementary Material for ``Efficient \\ Quantum Circuits for the Hilbert Transform''}
\author{Henry Zhang and Joseph Li}

\IEEEaftertitletext{\vspace{-1.05\baselineskip}}
\IEEEsecondmaketitle
\thispagestyle{headings}

\section{Proofs of Lemmas 1 \& 2}

\begin{lemma}
Let $f \in \mathbb R^{N\times\cdots\times N}$ be a real order-$d$ tensor. There exists a quantum circuit which maps $|f\rangle \to |\mathcal H [f]\rangle$ with circuit size $O(dn^2)$ and depth $O(dn)$, using $dn+1$ qubits.
\end{lemma}
\begin{proof}
For steps (i) and (iv), each $n$-qubit QFT can be implemented using $O(n^2)$ one- and two-qubit gates with depth $O(n)$ \cite{clevewatrous2000QFT, Camps2020QFT}. The $d$-dimensional QFT is separable into $d$ parallel $n$-qubit QFTs, so the total size sums to $O(dn^2)$ and combined depth remains $O(n)$. For step (ii), each $n$-controlled-X can be synthesized with $O(n^2)$ elementary gates and depth $O(n)$ without additional ancillae \cite{daSilva2022MCX, Barenco1995Gates}. We use exactly $d$ open-MCX gates to flip the ancilla, which contributes $O(dn^2)$ and $O(dn)$ depth in total. Finally, step (iii) only adds $d$ Pauli-Zs with $O(1)$ depth, which is subsumed.

The overall circuit size is hence $O(dn^2)$ and the overall depth is $O(dn)$, using $dn+1$ qubits. When mid-circuit-measurements are not hardware-supported, we cannot reuse the ancilla, so the algorithm requires $dn+d$ qubits. However, asymptotic size and depth do not change. See Figure \ref{fig:circuit}.
\end{proof}

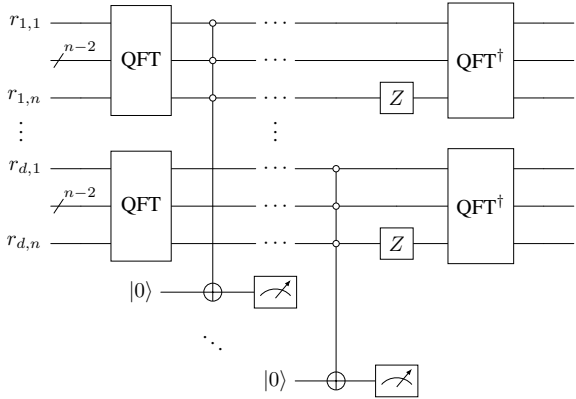
\begin{figure}[h]
    \begin{center}
    \hspace*{-.8em}
    \scalebox{0.8}{
        %\hspace*{2em}
        \begin{quantikz}[classical gap=0.06cm, row sep=0.6em, wire types={q,q,q,n,q,q,q,n,n,n}, thin lines]
            \lstick{$r_{1, 1}$} && \gate[3]{\kern0.08em\text{QFT}\kern0.08em} & \octrl{1} &\ \cdots \ &&&\gate[3]{\text{QFT}^\dagger}&&\\
            &\qwbundle[style={xshift=-1.2mm}]{n-2}&& \octrl{1} & \ \cdots \ &&&&& \\
            \lstick{$r_{1, n}$} &&& \octrl{5} & \ \cdots \ &&\gate{Z}&&& \\
            \lstick{\qvdots\kern0.5em} & & & & \push{\qvdots}&&& \\
            \lstick{$r_{d, 1}$} && \gate[3]{\kern0.08em\text{QFT}\kern0.08em} && \ \cdots \ &\octrl{1} && \gate[3]{\text{QFT}^\dagger} && \\
            &\qwbundle[style={xshift=-1.2mm}]{n-2} &&& \ \cdots \ &\octrl{1} &&&&\\
            \lstick{$r_{d, n}$} &&&& \ \cdots \ &\octrl{3} & \gate{Z} &&&\\
            && \ \push{{|0\rangle}}\ & \targ{} \setwiretype{q} & \meter{} \\
            &&& \qddots \\
            &&&& \ \push{{|0\rangle}}\ & \targ{} \setwiretype{q} &  \meter{}
        \end{quantikz}
    }
    \end{center}
    \caption{A simple implementation of the $d$-dimensional quantum Hilbert transform when dynamic circuits are not supported by hardware; $d$ ancillae are used to avoid increased depth.}
\label{fig:circuit}
\end{figure}

\begin{lemma}
Let $|f\rangle \in \mathbb C^{N^d}$ have Fourier coefficients $\hat{f}_{j_1,\ldots,j_d}$ defined in Equation $\textnormal{(8)}$ and spectral energies $S_{j_1,\ldots,j_d} = |\hat{f}_{j_1,\ldots,j_d}|^2$. Suppose the DC component satisfies
\begin{equation}
    \sum_{\substack{(j_1,\dots,j_d) \\ \exists m:j_m =0}}S_{j_1,\ldots,j_d} = p\sum_{j_1,\ldots,j_d=0}^{N-1} S_{j_1,\ldots,j_d}
\end{equation}
for some $p\in[0, 1]$. Then, $\mathbb{P}[\textnormal{success}] = 1-p$.
\end{lemma}
\renewcommand{\theequation}{S\arabic{equation}}
\setcounter{equation}{0}
\begin{proof}
After the $d$-dimensional QFT, the probability mass on a frequency $(j_1,\ldots,j_d)$ is $|N^{-d/2}\hat f_{j_1,\ldots,j_d}|^2
= N^{-d}S_{j_1,\ldots,j_d}$. Step (ii) is equivalent to applying
\begin{equation}
    \bigotimes_{r=1}^d \bigl(I_r - \Pi^{(r)}_{0}\bigr)
    = \sum_{\substack{(j_1,\dots,j_d)\\ \forall r: j_r\neq 0}}
    |j_1\cdots j_d\rangle\langle j_1\cdots j_d|,
\end{equation}
the projector onto the subspace spanned by those $|j_1\cdots j_d\rangle$ with no DC component. The success probability is thus
\begin{equation}
\sum_{\substack{(j_1,\ldots,j_d)\\ \forall r : j_r\neq 0}}
\frac{1}{N^d}S_{j_1,\ldots,j_d}
=1-\sum_{\substack{(j_1,\ldots,j_d)\\ \exists r : j_r=0}}
\frac{1}{N^d}S_{j_1,\ldots,j_d}.
\end{equation}
By the hypothesis, the spectral energy of the DC component is a $p$ proportion of the total energy, and Parseval's theorem gives $\sum_{j_1,\ldots,j_d}S_{j_1,\ldots,j_d}=N^d$. That
\begin{equation}
\mathbb P[\textnormal{success}]=1-\frac{1}{N^d}\cdot p\sum_{j_1,\ldots,j_d}S_{j_1,\ldots,j_d}
=1-p
\end{equation}
immediately follows.
\end{proof}

\section{Code Availability}
The code to replicate our figures and numerics is available at \url{https://github.com/henryzhng/quantum-hilbert-transform}.

\end{document}